\documentclass[10pt]{article}
\usepackage{amsmath,amssymb,mathtools}
\usepackage{soul}
\usepackage{graphicx,color}
\usepackage{xspace}
\usepackage{todonotes}
\usepackage{csquotes}
\usepackage{booktabs} 
\usepackage{tabularx} 

\usepackage{tikz,authblk}
\usepackage[small]{complexity}

\usepackage{hyphenat}

\newcommand{\basExpIOB}{ 3.41^k}
\newcommand{\basPolyIOB}{3.86^k}
\newcommand{\expIOB}{$\OO^*(\basExpIOB)$}
\newcommand{\polyIOB}{$\OO^*(\basPolyIOB)$}
\newcommand{\expColor}{$\OO^*(4.32^k)$}
\newcommand{\polyColor}{$\OO^*((2e)^{k+o(k)})$}

\usepackage{xpunctuate}
\usepackage[all,british]{foreign} 
\redefnotforeign[ie]{i.e\xperiodafter}
\redefnotforeign[eg]{e.g\xperiodafter}

\usepackage{tikz}
\newcommand*\circled[1]{\protect\tikz[baseline=(char.base)]{
            \protect\node[shape=circle,draw,inner sep=1.5pt] (char) {\normalfont\small #1};}}

\DeclarePairedDelimiter\ceil{\lceil}{\rceil}
\DeclarePairedDelimiter\floor{\lfloor}{\rfloor}

\def\any{\mathord{\color{black!33}\bullet}}%

\usepackage{stmaryrd}
\def\Iver#1{ \llbracket #1 \rrbracket }

\usepackage{accents}
\newlength{\dhatheight}
\newcommand{\doublehat}[1]{%
    \settoheight{\dhatheight}{\ensuremath{\hat{#1}}}%
    \addtolength{\dhatheight}{-0.35ex}%
    \hat{\vphantom{\rule{1pt}{\dhatheight}}%
    \smash{\hat{#1}}}}

\renewcommand{\le}{\leqslant}
\renewcommand{\leq}{\leqslant}
\renewcommand{\ge}{\geqslant}
\renewcommand{\geq}{\geqslant}

\renewcommand{\epsilon}{\varepsilon}

\def\tradeoff{\tau}

\newenvironment{tightcenter}
 {\parskip=0pt\par\nopagebreak\centering}
 {\par\noindent\ignorespacesafterend}

\usepackage{tikz}
\usetikzlibrary{calc}
\usepackage{xspace}
\usepackage{xargs}
\usepackage{xifthen}

\usepackage{framed}

\usepackage{ctable}
\newlength{\RoundedBoxWidth}
\newsavebox{\GrayRoundedBox}
\newenvironment{GrayBox}[1]%
   {\setlength{\RoundedBoxWidth}{\textwidth-4.5ex}
    \def\boxheading{#1}
    \begin{lrbox}{\GrayRoundedBox}
       \begin{minipage}{\RoundedBoxWidth}%
   }{%
       \end{minipage}
    \end{lrbox}%
    \begin{tightcenter}%
    \begin{tikzpicture}%
       \node(Text)[draw=black!20,fill=white,rounded corners,%
             inner sep=2ex,text width=\RoundedBoxWidth]%
             {\usebox{\GrayRoundedBox}};
        \coordinate(x) at (current bounding box.north west);
        \node [draw=white,rectangle,inner sep=3pt,anchor=north west,fill=white] 
        at ($(x)+(6pt,.75em)$) {\boxheading};
    \end{tikzpicture}
    \end{tightcenter}\vspace{0pt}%
    \ignorespacesafterend
}    

\newenvironment{problem}[2][]{\noindent\ignorespaces%
                                \FrameSep=6pt%
                                \parindent=0pt%
                \vspace*{-.5em}
                \ifthenelse{\isempty{#1}}{%
                  \begin{GrayBox}{\textsc{#2}}%
                }{%
                  \begin{GrayBox}{\textsc{#2} parametrised by~{#1}}%
                }
                \newcommand\Prob{Problem:}%
                \newcommand\Input{Input:}%
                \begin{tabular*}{\textwidth}{@{\hspace{.1em}} >{\itshape} p{1.6cm} p{0.8\textwidth} @{}}%
            }{
                \end{tabular*}%
                \end{GrayBox}%
                \vspace*{-.5em}
                \ignorespacesafterend
            }

\usepackage[amsmath,amsthm,thmmarks]{ntheorem}
\theoremseparator{.}
\newtheorem{theorem}{Theorem}

\newtheorem{lemma}[theorem]{Lemma}
\newtheorem{corollary}[theorem]{Corollary}
\newtheorem{proposition}[theorem]{Proposition}
\newtheorem{definition}[theorem]{Definition}
\newtheorem*{claim}{Claim}
\newcommand{\qedhere}{\ifmmode\qed\else\hfill\proofSymbol\fi}

\newcommand{\coef}{{\rm coef}}
\newcommand{\OO}{{\cal O}}
\DeclareMathOperator{\pf}{pf}
\def\alphaopt{\alpha^\star}

\def\Bjorklund{Bj{\"o}rklund\xspace}
\def\Wahlstrom{Wahlstr{\"o}m\xspace}

\usepackage{algorithm}
\usepackage{algorithmic}

\def\plog{\log^{\kern-.1pt{\scriptscriptstyle O(1)}}\kern-2pt}

\newcommand{\GG}[1]{#1}
\newcommand{\FR}[1]{#1}
\newcommand{\MW}[1]{#1}
\newcommand{\MZ}[1]{#1}

\begin{document}

\title{\GG{Designing Deterministic Polynomial-Space Algorithms by Color-Coding Multivariate Polynomials}\thanks{Gutin was partially supported by Royal Society Wolfson Research Merit Award, Reidl and Wahlstr{\"o}m by EPSRC grant EP/P007228/1, and Zehavi by ERC Grant Agreement no.~306992.}}

\author[1]{Gregory Gutin}
\author[1]{Felix Reidl}
\author[1]{Magnus Wahlstr{\"o}m}
\author[2]{Meirav Zehavi}
\affil[1]{Royal Holloway, University of London, TW20 0EX, UK}
\affil[2]{University of Bergen, Norway}
\date{}
\maketitle

\begin{abstract}
In recent years, several powerful techniques have been developed to design {\em randomized} polynomial-space parameterized algorithms. 
In this paper, we introduce an enhancement of color coding to design deterministic polynomial-space parameterized algorithms. Our approach aims at reducing the number of random choices by exploiting the special structure of a solution. Using our approach,  we derive polynomial-space \polyIOB-time (exponential-space \expIOB-time) deterministic
algorithm for {\sc $k$-Internal Out-Branching}, improving upon the previously fastest \emph{expo\-nential-space} $\OO^*(5.14^k)$-time algorithm for this problem. (The notation $\OO^*$ hides factors polynomial in the input size.)
We also design polynomial-space \polyColor-time (exponential-space \expColor-time)
  deterministic    algorithm for {\sc $k$-Colorful Out-Branching} on arc-colored digraphs and {\sc $k$-Colorful Perfect Matching} on planar edge-colored graphs.
In {\sc $k$-Colorful Out-Branching}, given an arc-colored digraph $D$, decide whether $D$ has an out-branching with arcs of at least $k$ colors. 
In {\sc $k$-Colorful Perfect Matching}, given an undirected graph $G$, decide whether $G$ has a perfect matching with edges of at least $k$ colors.
To obtain our polynomial space algorithms, we show that $(n,k,\alpha
k)$-splitters ($\alpha\ge 1$) and in particular $(n,k)$-perfect hash
families can be enumerated one by one with polynomial delay using polynomial space.
\end{abstract}

\section{Introduction}\label{sec:intro}

\noindent
\MZ{In this paper, we modify color coding to treat multivariate polynomials, and thus design an improved deterministic polynomial-space algorithm for {\sc $k$-Internal Out-Branching} ($k$-IOB). Before we elaborate on this problem and our contribution, let us first review related previous works that motivate our study.}
In recent years, several powerful algebraic techniques have been developed to
design {\em randomized} polynomial-space parameterized algorithms. The first
approach was introduced by Koutis \cite{AlgebraicPacking},
strengthened by Williams \cite{kPath2k}, and is nowadays
known as the multilinearity detection technique 
\cite{GroupAlgebrasFPT}. Roughly speaking, an application of this
technique consists of reducing the problem at hand to one where the objective
is to decide whether a given polynomial has a multilinear monomial, and then
employing an algorithm for the latter problem as a black box. \looseness-1

One of the huge breakthroughs brought about by this line of research
was \Bjorklund's~\cite{HamiltonDeterminant} proof that {\sc
Hamiltonian Path} is solvable in time $\OO^*(1.66^n)$ by a randomized
algorithm, improving upon the 50 year old~$\OO^*(2^n)$-time\footnote{%
  We use the common notation $\OO^*$ to hide factors polynomial in the input size.
} algorithm~\cite{TSPDP}. The existence of a
\emph{deterministic} $\OO^*((2{-}\epsilon)^n)$-time algorithm for {\sc Hamiltonian
Path}\MZ{, for a fixed $\epsilon>0$,} is still a major open problem. \FR{Further,} \GG{\Bjorklund's result is on undirected graphs\MZ{,} and the existence of an $\OO^*((2{-}\epsilon)^n)$-time algorithm for {\sc Hamiltonian
Path} on digraphs\MZ{, for a fixed $\epsilon>0$,} is another interesting open problem. }

Shortly afterwards, \Bjorklund \etal~\cite{NarrowSieves} have
transformed the ideas in~\cite{HamiltonDeterminant} into a
powerful technique to design randomized polynomial-space algorithms, referred
to as \emph{narrow sieves}. This technique is also based on the analysis of
polynomials, but it is applied quite differently. Here one
associates a monomial with each ``potential solution'' in such a way
that actual solutions correspond to unique monomials while incorrect
solutions appear in pairs. Thus, the polynomial summing these monomials, when
evaluated over a field of characteristic 2, is not identically 0 if and only
if the input instance of the problem at hand is a yes-instance. In this
context, the relevance of the Matrix Tree Theorem was already noted by Gabizon
\etal~\cite{RelaxedDisjointness}.

The narrow sieves technique, proven to be of wide applicability on its own,
later branched into several new methods. The one most relevant to our
study was developed by \Bjorklund
\etal~\cite{SCycles} and was translated into the language of
determinants by \Wahlstrom \cite{Wahl2013}. Here, the studied problem was
{\sc $S$-Cycle} (or {\sc $S$-Path}), where the goal is to determine whether an
input graph contains a cycle that passes through all the vertices of an input
set $S$ of size $k$. \Wahlstrom~\cite{Wahl2013} considered
a determinant-based polynomial (computed over a field of characteristic~2),
and analyzed whether there exists a monomial where the variable-set
representing $S$ is present. Very recently, \Bjorklund
\etal~\cite{ICALP17ToAppear} utilized the Matrix Tree Theorem to improve an 
FPT algorithm for \MZ{$k$-IOB}, where we are
asked to decide whether a given digraph has a  $k$-internal out-branching.
Recall that an {\em out-tree} $T$ is an orientation of a tree with only one
vertex of in-degree zero (called the {\em root}). A vertex of $T$ is a {\em
leaf} if its out-degree in $T$ is zero; non-leaves are called {\em internal
vertices}. An {\em out-branching} of a digraph $D$ is a spanning subgraph of
$D$, which is an out-tree, and an out-branching is {\em $k$-internal } if it hasat
least $k$ internal vertices.

\Bjorklund
\etal~\cite{ICALP17ToAppear} cleverly transformed $k$-IOB into a new problem, where
the goal is to decide whether a given polynomial (computed over a field of
characteristic $2$ to avoid subtractions) has a monomial with at least $k$
distinct variables. 

In this paper, we present an easy-to-use\footnote{In particular, no dynamic
programming/recursive algorithms are required.} \GG{modification of color coding} for designing
deterministic polynomial-space parameterized algorithms, inspired by the
principles underlying the above mentioned techniques. (A slight modification
of \GG{our approach} can be used to design faster, exponential-space algorithms,
but we believe that the main value of the approach is for polynomial-space
algorithms.) \GG{We will show that our approach brings significant speed-ups to algorithms for $k$-IOB.}
Roughly speaking, our
\GG{approach} can be applied as follows.

\begin{itemize}
  \item  Identify a polynomial such that it
  has a monomial with at least $k$ distinct variables (called a {\em witnessing
  monomial}) if and only if the input instance of the problem at hand is a yes-instance. 
  It should be possible to efficiently evaluate the polynomial (black
  box-access is sufficient here).

  \item  Color the variables of the polynomial with~$k$ colors using
  a polynomial-delay perfect hash-family. To improve the running time of
  this step, we apply a problem-specific \emph{coloring guide}
  to reduce the number of `random' colors. Given a $k$-coloring, we obtain
  a smaller polynomial by identifying all variables of the same color.

  \item  Use inclusion-exclusion to extract the coefficient of a 
  colorful monomial from the reduced polynomial. By the usual color-coding
  arguments, if \GG{the coefficient is not equal to zero} then the original polynomial contained
  a witnessing monomial.
\end{itemize}

\noindent
\GG{While we were unable to obtain non-trivial coloring guides to the following problems, even limited 
application of our approach is useful for designing polynomial-space algorithms for these problems.
It would be interesting to obtain non-trivial coloring guides for the problems.}

\paragraph{Colorful Out-Branchings and Matchings} 

\noindent
Every subgraph-search problem can be extended quite naturally by imposing
additional constraints on the solution, for example by letting the input graph
have labels or weights. One class of such constraints states that a required
subgraph of an edge-colored graph has to be {\em $k$-colorful}, \ie to
contain edges of at least $k$ colors.

One prominent problem is {\sc Rainbow Matching}\footnote{In the problem, given
an edge-colored graph $G$ and an integer $k$, the aim is to decide whether $G$
has a $k$-colorful matching of size $k$.}(also known as {\sc Multiple Choice
Matching}), defined in the classical book by Garey and Johnson
\cite{GareyJ79}. Itai \etal~\cite{ItaiRT78}
showed, already in 1978, that {\sc Rainbow Matching} is \NP-complete on
bipartite graphs. Three decades later, Le and Pfender~\cite{LeP14} revisited
this problem and showed that it is \NP-hard on several restricted graph
classes, which include (among others) paths, complete graph and $P_4$-free
bipartite graphs in which every color is used at most twice. 
Further examples of subgraph
problems with color constraints can be 
found in a survey by Mikio and Xueliang~\cite{ColorConstraintsSurvey}.
In this paper, we focus on two color-constrained problems: given an
edge-colored graph and an integer $k$, we ask for either a $k$-colorful spanning tree/outbranching
or a $k$-colorful perfect matching.

We first rely on the Matrix Tree Theorem to present a deterministic
polynomial\hyp{}space \polyColor-time (exponential-space \expColor-time)
algorithm for {\sc $k$-Colorful Out-Branching}, defined as follows:

\begin{problem}[k]{Colorful Out-Branching}
  \Input & A arc-colored digraph~$D$ and an integer~$k$  \\
  \Prob  & Does~$D$ have a $k$-colorful out-branching? 
\end{problem}

\noindent
We argue in Section \ref{sec:colOutBranch} that {\sc $k$-Colorful Out-Branching} is \NP-hard on various restricted graph classes such as cubic graphs.

Next, we rely on a Pfaffian computation to present a deterministic
polynomial-space \polyColor-time (exponential-space
\expColor-time) algorithm for {\sc $k$-Colorful Perfect Matching} on
planar graphs, defined as follows:

\begin{problem}[k]{Planar Colorful Perfect Matching}
  \Input & A planar edge-colored graph~$G$ and an integer~$k$  \\
  \Prob  & Does~$G$ have a $k$-colorful perfect matching? 
\end{problem}

\noindent
We will show in Section \ref{sec:colPerfMatch} by a simple reduction that
{\sc Planar $k$-Colorful Perfect Matching} is \NP-hard even on planar graphs
of pathwidth 2. It is worthwhile to note that while {\sc Rainbow Matching} can
be viewed as a special case of the well-known {\sc $3$-Set $k$-Packing}
problem, and in particular, a solution for {\sc Rainbow Matching} is small
(containing only $2k$ vertices and $k$ colors), the case of {\sc $k$-Colorful
Perfect Matching} is different in the sense that a solution is necessarily
large since not every $k$-colorful matching can be extended to a perfect
matching.

\paragraph{$k$-Internal Out-Branching} By utilizing the method of bounded
search trees on top of the above machinery, in Section
\ref{sec:maxInOutBranch}, we present a deterministic polynomial-space
\polyIOB-time (exponential-space \expIOB-time) algorithm for the problem {\sc
$k$-Internal Out-Branching ($k$-IOB)}, defined as follows:

\begin{problem}[k]{Internal Out-Branching (IOB)}
  \Input & A digraph~$D$ and an integer~$k$  \\
  \Prob  & Does~$D$ have a $k$-internal out-branching?
\end{problem}

\noindent
The undirected version of {\sc $k$-IOB},
called {\sc $k$-Internal Spanning Tree ($k$-IST)}, is defined similarly.

\begin{problem}[k]{Internal Spanning Tree (IST)}
  \Input & A graph~$G$ and an integer~$k$  \\
  \Prob  & Does~$G$ have a $k$-internal spanning tree? 
\end{problem}

\noindent
Note that $k$-IOB is a generalization of $k$-IST since the latter can easily
be reduced to $k$-IOB on symmetric digraphs, \ie digraphs in which every arc is
on a directed cycle of length 2. Since $k$-IST is \NP-hard ({\sc Hamiltonian
Path} appears as a special case for $k=n-2$) it follows that so is $k$-IOB.
The latter is, however, polynomial time solvable on acyclic
digraphs~\cite{GutinRK09}. While acyclic digraphs are precisely digraphs of
directed treewidth $0$, it turns out that $k$-IOB is \NP-hard
already for digraphs of directed \FR{treewidth}~1~\cite{DankelmannGK09}. By constrast, the {\sc Directed
Hamilton Path} problems is polynomial-time solvable on digraphs of directed
\FR{treewidth} $t$ \cite{JohnsonRST01} if $t$ is a constant. The $k$-IOB problem a
priori seems more difficult than the well-known {\sc $k$-Path} problem
(decide whether a given digraph has a path on $k$ vertices) since a witness of
a yes-instance of $k$-path is a subgraph of size $k$, that is, a path on $k$
vertices. However, it is easy to see that a witness of a yes-instance of
$k$-IOB (which has an out-branching) can be a subgraph of size $2k-1$, that is, an
out-tree with $k$ internal vertices and $k-1$ leaves. This simple but crucial
observation lies at the heart of previous algorithms for $k$-IOB.

\begin{table}[bht]
\centering%
\def\determ{~\emph{det}}%
\def\random{~\phantom{\emph{det}}~\emph{rand}}%
\def\poly{\emph{poly}}%
\def\exp{\phantom{\emph{poly}}~\emph{exp}}%
\def\direct{\emph{directed}}
\def\undirect{\emph{undirected}}
\begin{tabular}{llll>{~}l}
  \toprule
  Reference                                   & \multicolumn{1}{c}{Det./Rand.} 
                                                         & \multicolumn{1}{c}{Space} 
                                                                 & Graph           
                                                                               & \hspace*{-4pt}Time~$\OO^*(\cdot)$  \\ \midrule
  Prieto \etal~\cite{PrietoS05}               & \determ  & \poly & \undirect   & $2^{O(k\log k)}$ \\  
  Gutin \etal~\cite{GutinRK09}                & \determ  & \exp  & \direct     & $2^{O(k\log k)}$ \\
  Cohen \etal~\cite{kIOB49k}                  & \determ  & \exp  & \direct     & $55.8^k$         \\    
                                              & \random  & \poly & \direct     & $49.4^k$         \\  
  Fomin \etal~\cite{kIOB16k}                  & \determ  & \exp  & \direct     & $16^{k+o(k)}$    \\
                                              & \random  & \poly & \direct     & $16^{k+o(k)}$    \\ 
  Fomin \etal~\cite{kISP8k}                   & \determ  & \poly & \undirect   & $8^k$            \\  
  Shachnai \etal~\cite{RepFamUnified}         & \determ  & \exp  & \direct     & $6.855^k$        \\
  Daligault~\cite{DaligaultThesis}                   & \random  & \poly & \direct     & $4^k$            \\
  Li \etal~\cite{Li0CW17}                     & \determ  & \poly & \undirect   & $4^{k}$            \\
  Zehavi~\cite{Zehavi15}                      & \determ  & \exp  & \direct     & $5.139^k$        \\
                                              & \random  & \exp  & \direct     & $3.617^k$        \\    
   \Bjorklund \etal~\cite{SpottingTrees}   & \random  & \poly & \undirect   & $3.455^k$        \\
   \Bjorklund \etal~\cite{ICALP17ToAppear} & \random  & \poly & \undirect   & $2^k$            \\  \midrule
   Our work                                   & \determ  & \poly & \direct     & $\basPolyIOB$    \\
                                              & \determ  & \exp  & \direct     & $\basExpIOB$     \\
   \bottomrule
\end{tabular}\medskip
\caption{Previously known FPT algorithms for {\sc $k$-IOB} and {\sc $k$-IST}.}
\label{tab:knownresults}
\end{table}

\noindent
Parameterized algorithms for $k$-IST and $k$-IOB were first studied by Prieto
and Sloper \cite{PrietoS05} and Gutin \etal~\cite{GutinRK09}, who proved that
both problems are {\em fixed-parameter tractable (FPT)}, \ie admit
deterministic algorithms of running time $\OO^*(f(k))$, where $f(k)$ is an
arbitrary recursive function depending on the {\em parameter} $k$ only. Moreover,
both paper\MZ{s} showed that $f(k)=2^{O(k\log k)}.$ Since then several papers
improved complexities of deterministic and randomized algorithms for both
problems; we list these algorithms in Table~\ref{tab:knownresults}. We also
remark that approximation algorithms, exact exponential-time algorithms and
kernelization algorithms for both $k$-IOB and $k$-IST were extensively
studied, but the survey of such results is beyond the scope of this paper.

Our polynomial-space algorithm for {\sc $k$-IOB} is faster, in terms of $f(k)$, 
than not only the
previously fastest {\em exponential-space} $O^*(5.139^k)$-time deterministic
algorithm for $k$-IOB by Zehavi~\cite{Zehavi15}, but also the previously fastest
polynomial-space $\OO^*(4^{k})$-time deterministic algorithm for $k$-IST of Li
\etal~\cite{Li0CW17}. In contrast, 
it is not known  how to design a deterministic polynomial-space
$\OO^*((4-\epsilon)^k)$-time algorithm for the {\sc $k$-Path} problem. Indeed,
a deterministic polynomial-space \MZ{$\OO^*(4^{k+o(k)})$}-time algorithm for {\sc
$k$-Path} has been known since 2006~\cite{DivAndCol},\footnote{\MZ{Although Chen et al.~\cite{DivAndCol} do not explicitly prove that the space complexity of their deterministic algorithm can be made polynomial as well, this knowledge has become folklore.}} yet so far no
improvements without the help of exponential
space~\cite{RepFamEfficient,Zehavi15} have been made. 

The rest of the paper is organized as follows. In the next section, we
introduce necessary preliminary material. Section~\ref{sec:framework}
describes our approach in detail. The main contents of the next three
sections were discussed above. In Section~\ref{sec:polySpaceHash}, we prove a
derandomization theorem  forming part of our approach. 
In Section~\ref{sec:prop}, we show a proposition, which assists us in upper-bounding
running times of exponential-space algorithms. We conclude the paper in
Section~\ref{sec:conclusion} discussing some open problems.

\section{Preliminaries}\label{sec:prelims}

\noindent
We assume the reader is familiar with basic concepts and notations in Graph Theory and
Linear Algebra, and refer readers to textbooks in Graph Theory
\cite{jbj09,Diestel} and Linear Algebra
\cite{LayLinAlg} if additional details are required.

We will make use of the tighter version of Stirling's approximation
due to Robbins~\cite{TighterStirling}, which states that
\[
   \sqrt{2\pi k} \big( \frac{k}{e} \big)^k e^{1/(12k+1)} \leq k! \leq \sqrt{2\pi k} \big( \frac{k}{e} \big)^k e^{1/12k}.
\]
For~$\alpha \geq 1$ we will frequently use the function
\[
  \tradeoff(\alpha) = \Big( 1-\frac{1}{\alpha} \Big)^{\alpha-1} e,
\]
with the convention that~$\tradeoff(1)= e$. We will sometimes use 
the symbol~$\any$ in the following to denote a variable whose value
is arbitrary.

\paragraph{Operations on polynomials}
For a polynomial $P$ and a monomial $M$, we let $\coef_P(M)$ denote
the coefficient of $M$ in $P.$
For a polynomial~$P(x_1,\ldots,x_n)$ and subset of the variables~$C \subseteq \{x_1,\ldots,x_n\}$,
let~$P/C$ denote the polynomial obtained from~$P$ by replacing all variables in~$C$
by a new variable~$y_C$. We extend this notation to partitions
$\mathcal C := C_1 \uplus \ldots \uplus C_p$ and let~$P / \mathcal C$ denote
the polynomial~$(\ldots ((P / C_1) / C_2) \ldots ) / C_p$.

\paragraph{Structural Observation} It is not hard to decide whether a digraph
$D$ has an out-branching in linear time: $D$ contains an out-branching if and
only if $D$ has only one strongly connected component without incoming arcs
(see, \eg, \cite{jbj09}). Thus, in what follows, whenever we discuss problems
where a solution is in particular an out-branching, we will assume that the
digraph $D$ under consideration contains an out-branching.

\paragraph{Matrix Tree Theorem} In 1948 Tutte \cite{Tutte48} proved the
(Directed) Matrix Tree Theorem, which shows that the number of out-branchings
rooted at the same vertex $r$ in a digraph $D$ can be found efficiently by
calculating the determinant of a certain matrix derived from $D$. Here, we
require a generalization of this theorem, whose derivation from the original
theorem is folklore (a proof can be found, \eg, in \cite{ICALP17ToAppear}).

The {\em (symbolic) Kirchoff matrix} $K=K(D)$ of a directed multigraph $D$ on $n$ vertices 
is defined as follows, where we assume that the vertices are numbered 
from~$1$ to~$n$:
\begin{equation}
\label{Kdef}
K_{ij} = \left \{\begin{array}{ll}
\displaystyle{\sum_{\ell i\in A(D)}x_{\ell i}} & \mbox{if } i=j,\\
-x_{ij} &\mbox{ if } ij\in A(D),\\
0 & \mbox{ otherwise,}
\end{array}
\right.
\end{equation}
where $A(D)$ is the arc set of $D$.

In what follows, $[n]:=\{1,2,\ldots{},n\}$. For $i\in [n]$ we denote by
$K_{\bar{i}}(D)$ the matrix obtained from $K(D)$ by deleting the $i$th row and
the $i$th column. Moreover, let ${\cal B}_i$ denote the set of out-branchings
rooted at $i$. The following version of the (Directed) Matrix Tree Theorem
implies a natural one-to-one correspondence between the monomials of
$\det{}(K_{\bar{i}}(D))$ and the out-branchings in ${\cal B}_i$.

\begin{theorem}\label{thm:det-matrix-tree}
  For every directed multigraph $D$ with symbolic Kirchoff matrix $K(D)$ and
  $i\in V(D)$, $\det{}(K_{\bar{i}}(D)) = \displaystyle{\sum_{B\in{\cal
  B}_i}\prod_{ij\in A(B)}x_{ij}}$.
\end{theorem}

\paragraph{Planar Graphs, Perfect Matchings and Pfaffians} The Pfaffian is an important tool for polynomial-time counting algorithms,
closely related to perfect matchings of a graph. 
A square matrix $M \in \mathbb{R}^{n \times n}$ is \emph{skew-symmetric} if for every $i, j \in [n]$ it holds that $M(i,j) = -M(j,i)$. 
Let $M$ be skew-symmetric and of even dimension $2n$. For every partition of $[2n]$ into pairs $\{\{i_a,j_a\} \mid a \in [n]\}$, 
define a corresponding permutation $(i_1, j_1, \ldots, i_n, j_n)$ of $[2n]$ where $i_a<j_a$ for every $a \in [n]$ and $i_a < i_{a+1}$ for every $a \in [n-1]$.
Note that this is unique for every partition into pairs, and let $\Pi_n$ denote the set of permutation of partitions of $[2n]$.
Then the {\em Pfaffian} of $M$ can be defined as
\[
\pf(M) = \sum_{\pi \in \Pi_n} \sigma(\pi) \prod_{t=1}^n M(\pi(2t-1), \pi(2t)),
\]
where $\sigma(\pi)=  \pm 1$ is a sign term referred to as the \emph{signature}
of the permutation. The Pfaffian can be efficiently computed; in particular,
$\pf(M) = (\det(M))^2$. Among other applications, the Pfaffian can be used to
count the number of perfect matchings of a planar graph by computing an
orientation of the graph where every signature in the above sum is $+1$; see
Section~\ref{sec:colPerfMatch}.

\paragraph{Hash Families and Splitters} The notion of perfect hash family was employed
by Alon \etal~\cite{ColorCoding} when they introduced the
framework of color coding. Splitters are generalizations of perfect hash families; 
both are defined below. 

Let $\bar x=(x_1,\dots , x_n)\in [t]^n$ be a vector and $I=\{p_1,\dots
,p_{|I|}\}\subseteq [n]$, where $p_1<\dots <p_{|I|}$. Then $\bar
x[I]:=(x_{p_1}\dots ,x_{p_{|I|}})$. We say that $\bar x[I]$ {\em partitions $I$ almost
equally} if the cardinalities of the sets $\{i\in I:\ x_{i}=q\}$, $q\in [t]$
differ from each other by at most 1.
\begin{definition}[Splitter]\label{def:enrichedHash}
  An~\emph{$(n,k,t)$-splitter} is a family of vectors~$\mathcal S \subseteq [t]^n$
  such that for every index set~$I \in {[n] \choose k}$ there exists at least
  one vector~$\bar x \in \mathcal S$ such that~$\bar x[I]$ partitions~$I$
  almost equally. The \emph{size} of the splitter is~$|\mathcal S|$.
\end{definition}

\begin{definition}[Perfect hash family]
  An~$(n,k,k)$-splitter~$\mathcal S \subseteq [k]^n$ is called an
  \emph{$(n,k)$-perfect hash family}. In other words, for every index set~$I
  \in { [n] \choose k}$ there exists a vector~$\bar x \in \mathcal S$ such that~$\bar
  x[I]$ is a permutation of~$[k]$.
\end{definition}

\noindent
Note that the members of an $(n,k,t)$-splitter can equivalently be
interpreted as vectors from~$[t]^n$, as functions that map~$[n]$ into~$[t]$,
or as partitions of~$[n]$ into~$t$ blocks.

\section{Our Approach}\label{sec:framework}

\noindent
Let us now elaborate on the four main steps that constitute our approach.

\begin{itemize}
  \item[\circled{1}] \textbf{Polynomial Identification} \\
    Associate variables $X = \{x_1,x_2,\ldots,x_n\}$ with some elements (\eg,
    vertices or edges) of the input instance and identify a polynomial $P(x_1, \ldots ,
    x_n)$ over the reals that satisfies the following properties.
    \begin{itemize}
      \item $P(x_1, \ldots , x_n)$ has a monomial with at least $k$ distinct
            variables, called a {\em witnessing monomial}~$W$, if and only if the
            input instance is a yes-instance;
			\item All witnessing monomials of $P(x_1, \ldots , x_n)$ (in standard form) have non-negative coefficients;
      \item For any partition~$\mathcal C$ of~$X$ into~$k$ blocks and
            any assignment of the variables~$y_1,\ldots,y_k$ 
            \GG{($y_i$ corresponds to block $i$)},
            the polynomial $(P/{\mathcal C})(y_1,\ldots, y_k)$ can be evaluated efficiently
            using polynomial space.\footnote{The evaluation may use operations
            such as subtraction and division.}
    \end{itemize}

  \item[\circled{2}] \textbf{Derivation of Coloring Guides} \\
    Define a partition~$\mathcal S = S_1 \uplus S_2 \ldots S_p \uplus S_\bot$
    of the variables~$X$
    with the following property: if there exists a witnessing
    monomial~$W$, then in particular there exists such a monomial 
    with~$|W \cap S_i| = 1$ for~$i\in [p]$ and~$|W \cap S_\bot| = k - p$.
    We say that such a partition~$\cal S$ is a {\em coloring guide}.
    Note that we might consider more than one guide, \eg by a suitable
    branching procedure. In this case, we apply the following steps to
    all the generated guides and the above property needs to hold for
    at least one of the generated guides.

  \item[\circled{3}] \textbf{Color Coding \& Derandomization} \\
    Color the variables~$S_\bot$ with~$k-p$ colors uniformly at random and
    let the resulting color partition be~$\mathcal C := C_1 \uplus C_2 \uplus \ldots \uplus C_k$
    (where~$C_i = S_i$ for~$i \leq p$ and~$\bigcup_{i > p} C_i = S_\bot$).
    To derandomize this step, use an~$(n,k{-}p)$-perfect hash family~$\mathcal F$ that is enumerable
    with polynomial space (\cf Theorem~\ref{thm:polyDelayHash} stated shortly) to color~$S_\bot$.
    Proceed with the  next step for every coloring~$\mathcal C$.

  \item[\circled{4}] \textbf{Coefficient Extraction} \\ 
    Test whether~$P/{\mathcal C}$ contains a monomial $W$ in which all
    variables~$y_1,\ldots,y_k$ appear with a standard inclusion-exclusion
    algorithm. For example, we can apply Lemma~\ref{lem:coef} stated shortly for $P/{\mathcal C}$
    by  evaluating the corresponding $Q(y_1,\dots ,y_k)\neq 0$ at $y_i=1$ for all $i\in [k]$. Clearly, such a $W$ exists if and only if
    $Q(1,1,\dots ,1)\neq 0$.
    If so, conclude that~$P$ contains a witnessing monomial and 
    return that the instance is a yes-instance.
\end{itemize}

\noindent
We remark that Naor \etal~\cite{Splitters} proved that an $(n,k)$-perfect
hash family of size $\OO^*(e^{k+o(k)})$ can be computed in time
$\OO^*(e^{k+o(k)})$. They claimed that their construction can be
modified so that it is not required to compute the family ``at once'', but the
vectors in the family can be enumerated with polynomial delay. However, a
proof of the latter claim \FR{has not been published} (\GG{the proof was deferred to
a full version of that paper which never appeared}). 

In Section~\ref{sec:polySpaceHash}
we prove a more general theorem from which this claim can be derived as a corollary.
Importantly, our construction requires only polynomial space, a feature that
the construction by Naor \etal does not achieve.

\begin{theorem}\label{thm:polyDelayHash}
  There exists an $(n,k,\alpha k)$-splitter of size
  $k^{O(1)}  \tradeoff(\alpha)^{k + o(k)} \log n$ for every $\alpha \ge 1$.
  Moreover, the members of the family can be enumerated deterministically
  with polynomial delay using polynomial space.
\end{theorem}

\noindent
In case we are interested in polynomial space, we compute $\cal F$
simply as an $(n,k{-}p)$-perfect hash family. Otherwise, we compute $\cal F$
as an $(n,k{-}p,\alphaopt(k{-}p))$-splitter with a suitable constant~$\alphaopt
\approx \frac{4}{3}$. In the latter scenario, we can
run Steps~$\circled{3}$ and~$\circled{4}$ of our approach in time
$\OO^*(\tradeoff(\alphaopt)^{k-p + o(k)}2^{\alphaopt(k-p)+p})$. This requires
exponential space, as there are ${\alphaopt(k-p)+p \choose k}$ monomials for
which we need to check divisibility (see Lemma~\ref{lem:coef}), but for every
$I\subseteq [\alphaopt(k-p)+p]$, we would evaluate $P_{-I}$ only once rather
than once per such monomial. Then, we can rely on the following result proved
in Section \ref{sec:prop}.

\begin{proposition}\label{prop:enrichCalc}
  There exists a constant~$\alphaopt \approx \frac{4}{3}$ such that using an
  $(n,k-p,\alphaopt(k-p))$-splitter and exponential space, we can run Steps
  $\circled{3}$ and $\circled{4}$ of our approach in time
  $\OO^*(4.312^{k-p}2^p)$.
\end{proposition}
In the context of randomized algorithms, a slightly weaker result ($\OO^*(4.314^k)$) was obtained 
by H\"{u}ffner \etal~\cite{ColorEngineering}.

The following Lemma~\ref{lem:coef} provides a
way to extract from $P$ only monomials divided by a certain term, the crucial
operation in step~\circled{4}. The lemma is a modification of a lemma by Wahlstr{\"o}hm~\cite{Wahl2013}
for polynomials over a field of characteristic two.

\begin{lemma}[Monomial sieving]\label{lem:coef}
  Let $P(x_1, \ldots , x_n)$ be a polynomial over reals, and $J \subseteq [n]$ a
  set of indices. For a set $I \subseteq [n],$ define $P_{- I}(x_1,\ldots  ,
  x_n) = P(y_1,\ldots  , y_n),$ where $y_i = 0$ for $i \in I$ and $y_i = x_i$
  otherwise. Define
    \begin{equation}\label{qeq}
    Q(x_1,\ldots , x_n) =\sum_{I\subseteq J} (-1)^{|I|} P_{- I}(x_1,\ldots , x_n).
    \end{equation}
  Then, for any monomial $M$ divisible by $\Pi_{i\in J} x_i$, we have $\coef_Q(M) = \coef_P(M),$ and for
  every other monomial $M$ we have $\coef_Q(M) = 0.$  
\end{lemma}
\begin{proof}
  Consider a monomial $M$ with non-zero coefficient in $P.$ Observe first that
  for every $I \subseteq [n],$ we have $\coef_{P_{-I}}(M) = \coef_P (M)$ if no
  variable $x_i$ with $i \in I$ occurs in $M$, and $\coef_{P_{-I}}(M) = 0$,
  otherwise. Now, if $\Pi_{i\in J} x_i$ divides $M$, then out of the $2^{|J|}$
  evaluations for $I\subseteq J$, the monomial $M$ occurs in exactly one
  (namely, $I = \emptyset$). Thus,  $\coef_Q(M) = \coef_P(M).$

  If $\Pi_{i\in J} x_i$ does not divide $M$, \GG{note that} $J' = \{i \in J :
  x_i \mbox{ does not divide }M\},$ \GG{is nonempty} and observe that $\coef_{P_{-I}}(M) = \coef_P(M)$ for every $I \subseteq J'.$ Thus, sum (\ref{qeq}) for $M$ only is 
  $$\sum_{I\subseteq J}(-1)^{|I|} M=\sum_{I\subseteq J'}(-1)^{|I|} M=M\sum_{I\subseteq J'} (-1)^{|I|}=M(1-1)^{|J'|}=0.$$
  Applying the above results individually to every monomial in $P$ accounts
  for all occurrences of monomials in the sum defining $Q;$ the result follows. 
\end{proof}

\section{$k$-Internal Out-Branching}\label{sec:maxInOutBranch}

\noindent
In a digraph $D$, a {\em matching} is a collection of arcs without common
vertices. The following lemma establishes useful connections between matchings
and out-trees/out-branchings, which is the key to our computation of a
coloring guide.

\begin{lemma}\label{lemma:matching} \MZ{The following statements hold.} \leavevmode
\begin{itemize}
  \item[(i)] Let $T$ be an out-tree with $k\ge 0$ internal vertices. Then $T$ has a matching of size at least $k/2$.
  \item[(ii)] Let $D=(V,A)$ be a digraph containing an out-branching, and $M$ a matching in $D$. 
  Then, in polynomial time, we can find an out-branching of $D$ for which no arc of $M$ has both end-vertices as leaves.
\end{itemize}
\end{lemma} 
\begin{proof}
(1) We prove it by induction on $k\ge 0$. The claim obviously holds for $k=0$, so assume that $k\ge 1$. \GG{The {\em height} of a vertex $v$ in $T$ is the length of a longest path from $v$ to a leaf of $T$ reachable from $v$. Let $k_i$ be the number of vertices of $T$ of height $i$. }
Observe that $k_1\ge k_2$ and that $T$ has a matching $M_1$ with $k_1$ edges 
whose vertices are some leaves and all \GG{vertices of height 1}. Let $T'$ be an out-tree obtained from $T$ by deleting all leaves and \GG{vertices of height 1}. Observe that $T'$ has $k-k_1-k_2$ internal vertices and thus by induction hypothesis $T'$ has a matching $M_2$ of size at least $(k-k_1-k_2)/2\ge k/2 -k_1.$ Thus, the matching $M_1\cup M_2$ of $T$ is of size at least $k/2$. 

\smallskip

(2) Let $B$ be an out-branching of $D$ and suppose that both end-vertices of some arc $xy$ of $M$ are leaves in $B$. Then add $xy$ to $B$ and delete $zy$ from $B$, where $z$ is the in-neighbor of $y$ in $B$. In the resulting out-branching $B'$, $x$ is an internal vertex. Notice that $zy$ does not belong to $M$.
Hence, $B'$ contains 
one more arc of $M$ than $B$. 
Starting with an arbitrary out-branching and repeating the above exchange operation at most $|M|$ times, we will get an out-branching in which no arc of $M$ has both end-vertices as leaves. 
This process can be completed in polynomial time.
\end{proof}

\noindent
We now prove the main theorem of this section.
  
\begin{theorem}\label{thm:kIOB}
  There exists a deterministic polynomial-space \polyIOB-time
  algorithm for {\sc $k$-Internal Out-Branching}.
\end{theorem}

\begin{proof}
Let $D$ be a digraph, $V(D)=[n]$, and $M$ a maximum matching in $D$ of size
$t$. By Lemma \ref{lemma:matching}, we may assume that $k/2\le t\le k$ as otherwise we
can solve $k$-IOB on $D$ in polynomial time: \FR{For $t < k/2$, 
no out-tree with $k$ internal vertices exists and for~$t > k$, we can construct
a solution in polynomial time}.
Now we will follow our approach.
\smallskip

\noindent\circled{1}:
We associate one variable~$x_v$ with every vertex~$v \in [n]$. Replace every
variable $x_{ij}$ in (\ref{Kdef}) by $x_i$ and observe that now by
Theorem~\ref{thm:det-matrix-tree} we have that if the polynomial
$\det(K_\FR{{\bar r}}(D))$ over variables~$x_1,\ldots,x_n$ contains a monomial
with at least~$k$ different variables, then there exists an out-branching
rooted at~\FR{$r$} with at least~$k$ internal vertices. Note that for a
$k$-coloring~$\mathcal C$ of the variables we can evaluate the
polynomial~$\det(K_\FR{{\bar r}}(D))/ \mathcal C$ over
variables~$y_1,\ldots,y_k$ in polynomial time. We guess the root vertex~\FR{$r$}
and fix it for the following steps.
\smallskip

\noindent\circled{2}:
For every $c \in [k]\cup \{0\}$, we consider all sets $M'$ of $c$ edges
of $M$ in which both vertices are supposed to be internal vertices of some
$k$-internal out-branching (the edges of $M\setminus M'$ contain at least one
internal vertex in some $k$-internal out-branchings by Lemma \ref{lemma:matching} (ii)).

For the current choice of~$M'$, our coloring guide~$\mathcal S$ looks as follows.\footnote{Formally, we \MZ{somewhat} abuse terminology and notation for coloring guide here, but since the notions are very close,
we think it will not lead to a confusion.}
For every arc~$uv \in M'$, we add the sets~$\{u\}$ and~$\{v\}$ to~$\mathcal S$.
For every arc~$xy \in M \setminus M'$, we add the set~$\{x,y\}$. Finally, 
$S_\bot := V(D) \setminus V(M)$~contains all vertices outside of the matching.
\smallskip

\noindent\circled{3}: 
\FR{With~$t + c$ internal vertices of the out-tree in~$V(M)$ there are
$k-t-c$ vertices left to be located in~$V(D)\setminus V(M)$}.
Using an $(n,k-t-c)$-perfect hash family, we color the vertices of~$S_\bot$
and obtain a $k$-coloring $\mathcal C := C_1 \uplus C_2 \uplus \ldots \uplus C_{t+c} \uplus C_\bot$
of the variables~$X$ by combining the coloring from the hash family with
the coloring guide~$\mathcal S$.
\smallskip

\noindent\circled{4}:
We apply Lemma~\ref{lem:coef} to the polynomial~$\det(K_\FR{{\bar r}}(D)) /
\mathcal C$ over~$y_1,\ldots,y_k$~to search for a monomial that contains all
$k$ variables. If such a monomial exists, then~$\det(K_\FR{{\bar r}}(D))$ contains
a monomial with~$k$ distinct variables and we conclude
that~$D$ contains an out-branching rooted at~\FR{$r$} with $k$ or more internal
vertices. Otherwise, if we have not yet exhausted all colorings in the hash
family, we return to Step~\circled{3}, and if we have not yet exhausted all
guesses for $M'$, we return to Step~\circled{2}. If neither of these
conditions is true, we conclude that there is no $k$-internal out-branching
rooted at~\FR{$r$}.
\smallskip

\noindent 
Let us analyse the exponential part~$f(k)$ of the running time for the above steps. 
Step~\circled{2} to~\circled{4} take time at most
\GG{
\[
  \sum_{c=0}^{k-t}  {t\choose c} \tau(1)^{(k-t-c)(1+o(1))}2^k= \sum_{c=0}^{k-t} {t\choose c} e^{(k-t-c)(1+o(1))}2^k,
\]%
}%
using the~$(n,k-t-c)$-perfect hash family from Theorem~\ref{thm:polyDelayHash} (with $\alpha = 1$).
Consider the term~${t\choose c}e^{-(t+c)}$ and set $c = \beta t$, where $\beta \in (0,1)$. 
Using the well-known bound ${t \choose \beta t} \le 2^{tH(\beta)}$, where $H(\beta)=-\log (\beta^{\beta}(1-\beta)^{1-\beta})$,
we arrive at
\[
  {t\choose \beta t}e^{-t(1+\beta)} 
    \leq (\beta^{\beta}(1-\beta)^{1-\beta}e^{1+\beta})^{-t}\le \left(\frac{e+1}{e^2}\right)^t.
\]
\GG{
The second inequality above follows from the fact that $\min_{0<\beta<1}g(\beta)=\frac{e}{e+1}$, where $g(\beta)=\beta^{\beta}(1-\beta)^{1-\beta}e^{\beta}$, which can be verified by differentiating $g(\beta)$.
}
It is easy to verify that ${t\choose \beta t}e^{-t(1+\beta)}\le (\frac{e+1}{e^2})^t$ for $\beta=0$ and $\beta=1$ as well.
Therefore, 
\[
  \sum_{c=0}^{k-t} {t\choose c} e^{(k-t-c)(1+o(1))}2^k
    \leq k\left(\frac{e+1}{e^2}\right)^t (2e)^{k(1+o(1))}.
\] 
Observe that $\frac{e+1}{e^2}<1$ and thus $(\frac{e+1}{e^2})^t$ decreases
monotonically as $t$ grows. Therefore, $(\frac{e+1}{e^2})^t \leq (\frac{e+1}{e^2})^{k/2}$ 
as $t\ge k/2.$ Thus, $f(k)$ is bounded from above by
$k(4(e+1))^{k(1+o(1))/2} < k3.857^k$ for~$k$ large enough.
Since all of the above steps require polynomial space,
the algorithms requires polynomial space as well.
\end{proof}

\noindent
If $(D,k)$ is a positive instance of $k$-IOB, then we can find a $k$-internal
out-branching of $D$ also in polynomial space and time $O^*(3.857^k)$ by the
usual self-reducibility argument: Consider every arc $a$ of $D$ at a time
and remove it from $D$ if $(D-a,k)$ is a positive instance of $k$-IOB
until no arc can be removed. The non-empty remaining graph spans an out-branching
with~$k$ internal leaves.

Let us now see how our approach yields a faster algorithm if we use
exponential space.

\begin{theorem}
  There exists an exponential-space \expIOB-time
  algorithm for {\sc $k$-Internal Out-Branching}.
\end{theorem}

\noindent
The exponential-space strategy will use the exponential-space version of the
color-coding and sieving result, but we will in addition use \emph{truncated}
fast subset convolution to \MW{further speed up the algorithm, by rolling up
the $\binom{t}{c}$ guesses in Step~\circled{2} and the 
applications of Lemma~\ref{lem:coef} into a single exponential-space computation}. 
This was presented in Bj\"orklund \etal~\cite{BjorklundHKK10} (called trimmed Moebius inversion).

\begin{proof}
  We will present the proof as close to the structure of the proof of Theorem~\ref{thm:kIOB} as possible. 
  However, since the truncated fast subset convolution does not strictly speaking
  follow the coloring guide method, some discrepancies between the proofs are unavoidadble. 
  Let $D$ be a digraph with vertex set $V=[n]$. 
  
  \noindent \circled{1}:  
  As before, let $M$ a maximum matching in $D$ of size $t$ with $k/2\le t\le k$.
  Also guess a root vertex $r$ and keep it fixed through the following steps. 
  Number the edges of $M$ as $M=\{e_1, \ldots, e_t\}$. With every
  edge $e_i \in M$ we associate three variables $x_i$, $x_i'$, $x_i''$,
  and with every vertex $v \in V \setminus V(M)$ we associate a variable $x_v$.
  Further, for every edge $e_i = uv \in M$ define $x(u)=x_ix_i'$ and $x(v)=x_ix_i''$,
  and for every other vertex $v \in V \setminus V(M)$ define $x(v)=x_v$. 
  Replace every variable $x_{ij}$ in (\ref{Kdef}) by $x(i)$. 
  We will sieve the polynomial $\det(K_{\bar r}(D))$ according to
  two modes for every edge $e_i \in M$: Either only one endpoint
  of $e_i$ is internal in the out-branching, in which case
  we sieve for the variable $x_i$, or both are, in which
  case we sieve for both $x_i'$ and $x_i''$. Using 
  truncated fast subset convolution, we will sieve for
  these options in parallel.
  \smallskip
 
  \noindent \circled{2}+\circled{3}:
  Guess the number $c \in [k-t] \cup \{0\}$ of edges of $M$
  for which both endpoints are internal in the solution,
  but do not guess an explicit subset $M'$.
  Let~$k' := k-t-c$ be the number of internal vertices of the sought solution that 
  lie outside~$V(M)$, and recall the constant $\alphaopt$ from Proposition~\ref{prop:enrichCalc}.
  Further fix a coloring $f$ of the vertices of $V \setminus V(M)$
  into $\alphaopt \cdot k'$ colors, and assume that the $k'$
  internal vertices not present in $M$ all receive distinct colors by $f$. 
  Repeat the following steps for every choice of $f$ from an $(n, k', \alphaopt k')$-splitter. 
  \smallskip

  \noindent \circled{4}:
  We now describe the improved parallel sieving strategy.
  Define $X=\{x_i, x_i', x_i'' \mid i \in [t]\}$,
  and let $P(X, Y)$ be the result of replacing every variable 
  $x_j$, $j \in V \setminus V(M)$, in $\det(K_{\bar r}(D))$ by $y_{f(j)}$.
  Consider one particular choice $M' \subseteq M$, $|M'|=c$,
  of edges where both endpoints are assumed to be internal 
  vertices in the solution.
  We would then be seeking a monomial of $P(X,Y)$
  which contains both variables $x_i'$, $x_i''$ for every edge $e_i \in M'$,
  the variable $x_i$ for every edge $e_i \in (M \setminus M')$,
  and additionally variables of $k'$ further colors $y_{\ell} \in Y$.
  Combining this across all choices of $M'$ with $|M'|=c$, 
  we will thus want to run the sieving algorithm for the family of sets
  \[
  \mathcal{F}_c := \Big\{  \{x_i'\}_{i \in I} \cup \{x_i''\}_{i \in I} \cup \{x_i\}_{i \in [t]\setminus I} \,\Big|\, I \in \binom{[t]}{c} \Big\} \times \binom{[\alphaopt k']}{k'}.
  \]%
  For each $F \in \mathcal{F}_c$, let $Q_F(X,Y)$ be polynomial defined in Lemma~\ref{lem:coef}.
  Up to the choice of $c$ and $f$, the instance is positive if and only if there is
  some $F$ such that $Q_F(X, Y) \neq 0$, and we can use $X=\mathbf{1}$, $Y=\mathbf{1}$ for the evaluation, where $\mathbf{1}$ denotes a vector with all components equal 1. 
  
  Using the truncated fast subset convolution, we can compute all
  evaluations $Q_{F}(\mathbf{1},\mathbf{1})$ as above in time 
  proportional to the number of subsets of sets $F \in \mathcal{F}_c$, 
  up to a polynomial factor~\cite{BjorklundHKK10}.
  Concretely, let $\mathcal{I}_c=\{I \subseteq F \mid F \in \mathcal{F}_c\}$.
  Then the truncated fast subset convolution runs in time $\OO^*(|\mathcal{I}_c|)$,
  and we repeat this procedure for every choice $c$ and for every member $f$ of
  the $(n,p,\alphaopt k')$-splitter,

  \emph{Running time.} To simplify the analysis of the running time, we 
  write $\mathcal{J}_c=\{I \cap X \mid I \in \mathcal{I}_c\}$.
  Then $|\mathcal{I}_c| \leq |\mathcal{J}_c| \cdot 2^{\alphaopt k'}$,
  and the product of the size of the splitter and $2^{\alphaopt k'}$
  can be bounded as $\OO^*(4.312^{k'})$ as in Proposition~\ref{prop:enrichCalc}.
  It remains to bound $|\mathcal{J}_c|$. We further split $\mathcal{J}_c$
  into $c+1$ \emph{levels}, where a set $I \in \mathcal{J}_c$
  belongs to level $i$, $0 \leq i \leq c$, if there are exactly
  $i$ edges $e_j \in M$ such that $I \cap \{x_j', x_j''\} \neq \emptyset$.
  Thus, the number of sets at level $i$ of $\mathcal{J}_c$ equals
  \[
  \binom{t}{i} 3^i \sum_{j=0}^{t-c} \binom{t-i}{j} \leq \binom{t}{i} 3^i 2^{t-i},
  \]
  where 
  the factor of $3$ comes from the three options $\{x_j'\}$, $\{x_j''\}$, $\{x_j', x_j''\}$
  and the expression in the summation corresponds to whether $x_j \in I$
  for the remaining edges. Note that the upper bound is essentially tight assuming $t-c > (t-i)/2$. 
  Therefore, if we split out the total work per level, we get
  \begin{equation}\label{eq1}
    \max_{i \in [c]\cup \{0\}} \binom{t}{i} 3^i 2^{t-i} 4.312^{k-t-c} = 2^t4.312^{k-t-c} \max_{i \in [c]\cup \{0\}} \binom{t}{i} 1.5^i.
  \end{equation}
  To obtain an upper bound for (\ref{eq1}), we will use the following observation.
  
  \begin{claim}
     Let $b$ be a real, $t$ an integer, and $x$ an integral variable. The function
     $g(x)={t \choose x}b^x$ monotonically increases if and only if $x\le
     \frac{b(t+1)}{b+1}$.
  \end{claim}
  \begin{proof} It suffices to observe that $g(x)/g(x-1)\ge 1$ if and only if $x\le \frac{b(t+1)}{b+1}$. \end{proof}

  \noindent
  Furthermore, since $g(x+1), g(x-1) \leq g(x) \cdot bt$ for the function $g(x)$ defined in the claim,
  up to lower-order terms we need not be concerned with the precise value of $x$. 
  We now consider two cases. First, assume that $c \geq 0.6t$.
  Then by the claim above, we have
  \[
  \binom{t}{i} 1.5^i \leq t^{O(1)} \binom{t}{\lceil 0.6t\rceil} 1.5^{0.6t} \leq t^{O(1)} 2^{H(0.6)t} 1.5^{0.6t},
  \]
  where we used the well-known bound
  $
  \binom{t}{\alpha t} \leq 2^{H(\alpha) t}
  $
  for the binary entropy function $H(\alpha) = -\alpha \log \alpha - (1-\alpha)\log (1-\alpha)$.
  Inserting the exponential part into (\ref{eq1}), we get a bound of
  \[
  2^t 4.312^{k-t-0.6t} 2^{H(0.6)t} 1.5^{0.6t} <  4.312^k0.482484^t <3^k
  \]
 as $t \geq k/2$. Thus, we next consider $c<0.6t$. Then by the claim, we have
  \[
  \max_{i \in [c] \cup \{0\}} \binom{t}{i} 1.5^i = \binom{t}{c} 1.5^c.
  \]
  Inserting it into (\ref{eq1}) and rearranging, we get a bound of 
  \[
  2^t 4.312^{k-t} \binom{t}{c} (1.5/4.312)^c.
  \]
  Let  $a=1.5/4.312$ and $\alpha=a/(a+1)$.
  By the claim, 
  the latter half of the product above is maximized 
  around $c=\alpha t$ giving the following upper bound to the product:
  \[
  t^{O(1)} 2^t 4.312^{k-t} 2^{H(\alpha) t} a^{\alpha t} < t^{O(1)}4.312^k0.625172^t = \OO^*(3.41^k),
  \]
as $t \geq k/2$. Finally, since $\alpha < 0.6$, an integer around $\alpha t$ is a valid value of $c$.
\end{proof}
  
\noindent \MW{We remark that while the running time bound is
  essentially tight for the algorithm we describe,
  we cannot exclude that there is a more efficient way of
  implementing the sieving, e.g., using less than three
  variables per edge in $M$.}
  
 \section{$k$-Colorful Out-Branching}\label{sec:colOutBranch}

\noindent
First, let us argue that {\sc $k$-Colorful Out-Branching} is \NP-hard. To this
end, we have a simple reduction from {\sc Hamilton Path}: in a digraph $D$
define the color of every arc outgoing from $v\in V(D)$ to be $c_v$. Clearly,
$D$ has a {\sc Hamilton Path} if and only if it has a $k$-colorful 
out-branching with $k=|V(D)|-1$. Thus, on any graph class where {\sc Hamilton
Path} is \NP-hard, {\sc $k$-Colorful Out-Branching} is \NP-hard as well.
In fact, we can have the same simple reduction from $k$-IOB to
{\sc $k$-Colorful Out-Branching}, which shows that the latter generalizes the
former.

We follow the structure of the proof of Theorem \ref{thm:kIOB} to show
the following result for the colorful variant.

\begin{theorem}\label{thm:colorfulOutBranch}
  There exists a deterministic polynomial-space $\OO^*((2e)^{k+o(k)})$-time
  (and exponential-space $\OO^*(4.312^{k}$)-time) algorithm for 
  {\sc $k$-Colorful Out-Branching}.
\end{theorem}

\begin{proof}
  Let $D=(V,A)$ be a directed multigraph, and let $c(a)$ be the color of $a$
  for every arc $a \in A$. We guess the root vertex~\FR{$r$} and fix it for the
  following steps. Let $C$ be the full set of colors used, and without loss of
  generality let $C=[t]$, $t \geq k$, and $A=\{a_1,\ldots,a_m\}$. Create a set
  $Z=\{z_1, \ldots, z_m\}$ of corresponding variables, and define the
  polynomial $P'(z_1,\ldots,z_m) = \det(K_\FR{{\bar r}}(D))$. By 
  Theorem~\ref{thm:det-matrix-tree}, our problem has now been reduced to determining whether
  there is a monomial $M$ in $P'$ such that the corresponding out-branching
  contains arcs of at least $k$ different colors. We solve this problem
  according to our approach as follows. Here, we first describe the
  polynomial-space algorithm, and later explain how it can be sped-up at the
  cost of exponential space.
  \smallskip
  	
  \noindent\circled{1}: Create a set $X=\{x_1, \ldots, x_t\}$ of variables
  corresponding to colors. We define the polynomial $P$ by
  $P(x_1,\ldots,x_t)=P'(x_{col(1)},\ldots,x_{col(m)})$, where $col(i)=c(e_i)$
  for all $i\in[m]$. Then, as we argued above, $D$ has a $k$-colorful out-branching 
  rooted at \FR{$r$} if and only if $P$ has a witnessing monomial.
  Moreover, note that for a $k$-coloring~$\mathcal C$ of the variables, we can
  evaluate the polynomial~$P/ \mathcal C$ over variables~$y_1,\ldots,y_k$ in
  polynomial time, as this can be done by evaluating~$P'$.
  \smallskip

  \noindent\circled{2}: In this application, our coloring guide is empty
  ($S_\bot = V(D)$). Thus, it can be ignored in the following steps.
  \smallskip

  \noindent\circled{3}: Using a $(t,k)$-perfect hash family, we recolor the colors of~$C$
  and obtain a $k$-coloring $\mathcal C := C_1 \uplus C_2 \uplus \ldots \uplus
  C_{k}$ of the variables~$X$.
  \smallskip

  \noindent\circled{4}: We apply Lemma~\ref{lem:coef} to the polynomial~$P / \mathcal C$
  over~$y_1,\ldots,y_k$~to search for a monomial that contains all $k$
  variables. If such a monomial exists, then~$P$ contains a monomial with~$k$
  or more distinct variables and we conclude that~$D$ contains an out-branching rooted at~\FR{$r$} with $k$ 
  or more colors. Otherwise, if we have not
  yet exhausted all colorings in the hash family, we return to
  Step~\circled{3}, and otherwise we conclude that the there is no
  $k$-colorful out-branching rooted at~\FR{$r$}. \looseness-1

  \smallskip
  \noindent 
  Let us now analyse the exponential part~$f(k)$ of the running time
  for the above steps. By Theorem~\ref{thm:polyDelayHash} there are
  $\OO^*(e^{k+o(k)})$ colorings in our $(t,k)$-perfect hash family, and by
  Lemma~\ref{lem:coef} the sieving can be performed in time $\OO^*(2^k)$ for
  every individual coloring. Hence, the total time is $\OO^*((2e)^{k+o(k)})$
  as claimed.

  Finally, we consider the case where we may use exponential space. Then, we
  use a $(t,k,\alphaopt k)$-splitter rather than a $(t,k)$-perfect hash
  family. Thus, by Proposition~\ref{prop:enrichCalc}, the total running time
  is $\OO^*(4.312^k)$, as claimed.
\end{proof}

\noindent
We can improve the above result if the input graph is colored with
\emph{exactly}~$k$ colors: In Step~\circled{2} of the above proof we
do not use an empty coloring guide, but instead a partition of $A$ such that
two arcs are in the same block of the partition if and only if they have the
same color.  Notice that in this case, in Step \circled{3} we construct a
$(t,0)$-perfect hash family, which means that every color simply retains its
original color. Thus, we derive the following corollary.

\begin{corollary}\label{corollary:colorfulOutBranch}
  There exists a deterministic polynomial-space $\OO^*(2^k)$-time algorithm for
  {\sc $k$-Colorful Out-Branching} on directed multigraph with exactly $k$
  colors.
\end{corollary}

\section{$k$-Colorful Perfect Matching on Planar Graphs}\label{sec:colPerfMatch}

\noindent
Let us first show the \NP-hardness of {\sc $k$-Colorful Perfect Matching}.

\begin{lemma}
  {\sc $k$-Colorful Perfect Matching} is \NP-hard on planar graphs of pathwidth 2.
\end{lemma}

\begin{proof}
  We present a reduction from {\sc Rainbow Matching} on paths, which is known to
  be \NP-hard~\cite{LeP14}. To this end, let $(P,k)$ be an instance of {\sc
  Rainbow Matching} on paths. Denote $P=v_1-v_2-\cdots-v_n$. Now, we construct
  an instance $(G,k+1)$ of {\sc $k$-Colorful Perfect Matching} as follows. Let
  $c^\star$ be a new color, and define $P'=v'_1-v'_2-\cdots-v'_n$ as a copy of
  $P$ where the color of each edge is $c^\star$. Then, $G$ is defined as the graph
  obtained from $P$ and $P'$ by adding the edge $v_i-v'_i$ of color $c^\star$
  for all $i\in[n]$. Clearly, $G$ is a planar graph of pathwidth 2.

  On the one hand, suppose that $(P,k)$ has a $k$-colorful matching $M$ of size
  $k$. Denote $I=\{i\in[n-1]: v_iv_{i+1}\in M\}$ and $I^+=\{i+1: i\in I\}.$ Then, $M'=M\cup
  \{v'_iv'_{i+1}: i\in I\}\cup\{v_iv'_i: i\notin I\cup I^+\}$ is a $(k+1)$-colorful
  perfect matching in $G$. On the other hand, suppose that $(G,k+1)$ has a
  $(k+1)$-colorful perfect matching $M$. Since the color of all edges outside
  $P$ is $c^\star$, the set of edges of $M$ that belong to $P$ is
  $k$-colorful. In particular, this set contains a $k$-colorful matching of
  size $k$.
\end{proof}

\noindent
For our algorithm, we will need the following famous result
which implies that planar perfect matchings
can be counted in polynomial time. We will need the following 
more detailed version of the original result:

\begin{theorem}[\cite{Kasteleyn67PM}]
  \label{thm:planarpfaffian}
  Let $G=(V,E)$ be a planar graph, and let $X=\{x_e \mid e \in E\}$
  be a collection of real-valued variables.
  There is a polynomial-time algorithm that produces a matrix $A$ 
  over $X$ such that
  \[
  \pf(A) = \sum_{M \in \mathcal{M}} \prod_{e \in M} x_e,
  \]
  where $\mathcal{M}$ ranges over all perfect matchings of $G$. 
\end{theorem}

\noindent
Using our approach, we can now easily derive the following theorem in a manner very similar to the one used to prove Theorem \ref{thm:colorfulOutBranch}. For the sake of completeness, we present the required details.

\begin{theorem}
  {\sc $k$-Colorful Perfect Matching} on planar graphs  can be solved
  by a deterministic   algorithm in time \polyColor\ with polynomial space,
  and in time \expColor\ with exponential space.
\end{theorem}
\begin{proof}
  Let $G=(V,E)$ be a planar graph, and let $c(e)$ be the color of $e$
  for every edge $e \in E$. Let $C$ be the full set of colors used,
  and without loss of generality let $C=[t]$, $t \geq k$,
  and $E=\{e_1,\ldots,e_m\}$. Create a set $Z=\{z_1, \ldots, z_m\}$
  of corresponding variables, and let $A$ be 
  the matrix computed by Theorem~\ref{thm:planarpfaffian}. 
  Define the polynomial $P'(z_1,\ldots,z_m) = \pf(A)$.
  Recall that $P$ can be evaluated in polynomial time.   Our problem has now been reduced to determining whether
  there is a monomial $M$ in $P'$ such that the corresponding 
  perfect matching contains edges of at least $k$ different colors. We solve this problem according to our approach as follows. Here, we first describe the polynomial-space algorithm, and later explain how it can be sped-up at the cost of exponential space.
\smallskip
	
\noindent\circled{1}: Create a set $X=\{x_1, \ldots, x_t\}$ of variables corresponding to colors. We define the polynomial $P$ by $P(x_1,\ldots,x_t)=P'(x_{col(1)},\ldots,x_{col(m)})$, where $col(i)=c(e_i)$ for all $i\in[m]$. Then, as we argued above, the input instance is a yes-instance if and only if $P$ has a witnessing monomial. Moreover, note that for 
a $k$-coloring~$\mathcal C$ of the variables, we can evaluate the polynomial~$P/ \mathcal C$
over variables~$y_1,\ldots,y_k$ in polynomial time, as this can be done by evaluating $P'$.
\smallskip

\noindent\circled{2}: In this application, our coloring guide is empty ($S_\bot = V(D)$). Thus, it can be ignored in the following steps.
\smallskip

\noindent\circled{3}: Using a $(t,k)$-perfect hash family, we recolor the colors of~$C$
and obtain a $k$-coloring $\mathcal C := C_1 \uplus C_2 \uplus \ldots \uplus C_{k}$
of the variables~$X$.
\smallskip

\noindent\circled{4}: We apply Lemma~\ref{lem:coef} to the polynomial~$P / \mathcal C$
over~$y_1,\ldots,y_k$~to search for a monomial that contains all $k$ variables. If
such a monomial exists, then~$P$ contains a monomial with~$k$ or
more distinct variables and we conclude that~$G$
contains a perfect matching with $k$ or more colors. Otherwise, if we have not yet exhausted all colorings in the hash family, we return to Step~\circled{3}, and otherwise we conclude that the input instance is a no-instance. 

\smallskip
\noindent Let us now analyse the exponential part~$f(k)$ of the running time for the above steps. By Theorem~\ref{thm:polyDelayHash} there are $\OO^*(e^{k+o(k)})$ colorings in our $(t,k)$-perfect hash family, and by Lemma~\ref{lem:coef}
  the sieving can be performed in time $\OO^*(2^k)$
  for every individual coloring. Hence, the total time is $\OO^*((2e)^{k+o(k)})$
  as claimed.

Finally, we consider the case where we may use exponential space. Then, we use
a $(t,k,\alphaopt k)$-splitter rather than a $(t,k)$-perfect hash family.
Thus, by Proposition~\ref{prop:enrichCalc}, the total running time is
$\OO^*(4.312^k)$ as claimed.
\end{proof}

\noindent
We conclude by observing that there is little hope to
apply our approach to {\sc Rainbow Matching}.
In particular, counting not necessarily perfect matchings 
in a planar graph is \#P-hard, so there is no 
plug-in replacement for Theorem~\ref{thm:planarpfaffian}
for general matchings.

\section{Enumerating Splitters with Polynomial Delay}\label{sec:polySpaceHash}

\def\init{\textsf{in}}
\def\query{\textsf{qr}}

\noindent 
In this section, we prove Theorem \ref{thm:polyDelayHash}. We essentially follow
the construction by Naor \etal~\cite{Splitters} while taking care to keep the
space consumption polynomial. In particular, the idea by Fomin \etal~\cite{RepFamEfficient}
(used in the context of construction representative sets)
to frame the construction in a way that makes a repeated application possible, turns
out to be a crucial component. To this end, we will need the following definition
of an \emph{indexed splitter} which treats splitter families as data structures that
enumerate vectors instead of fixed collections of vectors.

\begin{definition}[Indexed splitter]
  An \emph{indexed $(n,k,t)$-splitter} of size~$m$ is a data structure~$\mathcal S$ 
  that for~$i \in [m]$ returns a vector~$\bar x_i \in [t]^n$ such that
  $\{\bar x_i\}_{i\in[m]}$ is an $(n,k,t)$-splitter.
  The \emph{query time}~$t_\query(n,k,t)$ and~\emph{query space}~$s_\query(n,k,t)$ are the resources
  needed by~$\mathcal S$ to compute any such~$\bar x_i$, where we exclude the space needed
  by~$\bar x_i$ from~$s_\query$.
  The \emph{initialization time}~$t_\init(n,k,t)$ and \emph{initialization space}~$s_\init(n,k,t)$
  are the resources needed to compute~$\mathcal S$ given~$n,k$ and~$t$.
\end{definition}

\noindent
We will call the tuple \GG{$(m,t_\init,s_\init,t_\query,s_\query)$} the \emph{profile}
of an indexed splitter. Note that every splitter of size~$f(n,k,t)$ is 
an indexed splitter with query-time proportional to~$\log f(n,k,t)$:
we simply store the whole splitter in memory according to some (arbitrary)
order.

One of the main ingredients will be the following two splitters
constructed by Naor, Schulman, and Srinivasan:

\begin{proposition}[\cf \cite{Splitters}]\label{prop:famA}
  There exists an~$(n,k,k^2)$-splitter~$\mathcal A(n,k)$ of
  size $k^{O(1)} \log n$ that can be efficiently constructed using~$k^{O(1)} n \log n$ space.
\end{proposition}

\begin{proposition}[\cf \cite{Splitters}]\label{prop:famB}
  For all~$k, t \leq n$ there exists an indexed~$(n,k,t)$-splitter~$\mathcal B(n,k,t)$
  of size~${n \choose t-1}$ 
  with~$t_\init(n,k,t), s_\init(n,k,t) = O(t \log n)$,
  $t_\query(n,k,t) = O(n t \log n)$, and
  $s_\query(n,k,t) = O(t \log n)$.
\end{proposition}
\begin{proof}
  The underlying splitter corresponds to the set of all ordered tuples~$\bar i \in [n]^{t-1}$.
  Given such a tuple~$\bar i$, we assign all elements in the range
  $\big[0,\bar i[0]-1\big]$ the value~$0$, all elements in the range~$\big[\bar i[0], \bar i[1]-1\big]$
  the values~$1$, and so on. Clearly, every index set of size~$k$ is partitioned almost
  equally by at least one of these vectors. This construction is easily transformed
  into an indexed splitter with the claimed profile by choosing an appropriate 
  indexing of the ordered tuples from~$[n]^{t-1}$.
\end{proof}
  
\noindent
One further core component needed
here is a $k$-wise independent \emph{sampling space}, whose properties we can
use to generate a small $(n,k)$-perfect hash family. Let us recall the very
elegant construction by Joffe \cite{Joffe} for such a space:

\begin{definition}
  A probability space~$\Omega$ with~$n$ random variables~$\{X_i\}_{i \in [n]}$
  is \emph{$k$-wise independent} if for every index set~$I \subseteq [n]$ of
  size~$k$ the random variables~$\{X_i\}_{i \in I}$ are mutually independent.
\end{definition}

\begin{theorem}[Joffe~\cite{Joffe}]\label{thm:joffe}
  Let~$p$ be prime and~$k < p$. If~$\{X_i\}_{i\in[k]}$ are random variables
  uniformly distributed on~$\{0,\ldots,p-1\}$, then the variables~$\{Y_i\}_{i\in[p]}$
  defined via
  \[
    Y_i := ( X_1 + iX_2 + \ldots + i^{k-1} X_k ) \pmod p
  \]
  are uniformly distributed over~$\{0,\ldots,p-1\}$ and $k$-independent.
\end{theorem}

\noindent
Note that the variables~$\{Y_i\}_{i\in[p]}$ range over~$\{0,\ldots,p-1\}$ and not
over~$[k]$. However, because they are uniformly distributed, we can easily
`downsample' them without much loss.

\begin{lemma}\label{lemma:enumerable-sampling}
  Let~$p$ be prime and~$k \leq t < p$. There exists a $k$-wise independent
  sampling space~$H^\star_{p,k,t}$ for random variables~$\{\hat Y_i\}_{i \in
  [p]}$ over~$[t]$ where the $\hat Y_i$ are identically distributed and almost
  uniform in the sense that~$\big| \mathbb{P}[\hat Y_i = r] - \frac{1}{t} \big|
  \leq \frac{1}{p}$ for~$r \in [t]$.
  Moreover, the members of~$H^\star_{p,k,t}$ can be listed sequentially
  in time~$O(p^t k^2 \log p)$ by an algorithm using~$O(t \log p)$ bits.
\end{lemma}
\begin{proof}
  Let the variables~$\{X_i\}_{i\in[k]}$ and~$\{Y_j\}_{j\in[p]}$ be defined as above.
  We further define~$\hat Y_j := Y_j \pmod t$. Since the variables~$Y_j$ are
  $k$-wise independent, so are the variables~$\hat Y_j$. The distribution
  of each~$\hat Y_j$ follows immediately from the fact that
  \[
    \frac{\floor{p/t}}{p} \leq \mathbb{P}\!\big[\hat Y_j = r\big] \leq \frac{\ceil{p/t}}{p}
  \]
  for each~$r \in [t]$. To list all members of~$H^\star_{p,k,t}$, we enumerate all~$p^t$ possible
  assignments of~$\{X_i\}_{i\in[k]}$ and compute the values for~$\{\hat Y_i\}_{j\in[p]}$.
  We need~$O(t \log p)$ bits to store a counter for the first enumeration as
  well as~$O(k \log p)$ bits to enumerate the polynomials defining the~$\hat Y_i$.
  Since each polynomial can be evaluated using~$O(k^2)$ arithmetic operations, the
  claimed running time follows.
\end{proof}

\noindent
\FR{The following construction of a basic splitter follows the one by Noar \etal closely,
however, our rephrasing and analysis is more suitable for the final construction
of the indexed splitter.}

\begin{lemma}\label{lemma:new-splitter}
  For every~$k < \sqrt{n}$ and~$\alpha \geq 1$ there is an $(n,k,\alpha k)$-splitter~$\mathcal C_\alpha(n,k)$
  of size~$O\big( \tradeoff(\alpha)^k \, k \log n \big)$ \FR{that} can
  be constructed in $|\mathcal C_\alpha(n,k)| k^2 {n \choose k} (2n)^{\alpha k} \plog n$
  time using $O(|\mathcal C_\alpha(n,k)| \, \log n )$ space.
\end{lemma}
\begin{proof}
  Let~$n \leq p \leq 2n$ be the smallest prime at least as large as~$n$. We can
  identify~$p$ using the AKS test or any of its recent improvements in~$n \plog p$
  time. For ease of presentation, we will assume that~$\alpha k$ is an integer.
  Let~$H^\star = H^\star_{p,k,\alpha k}$ be the $k$-wise independent probability space defined 
  in Lemma~\ref{lemma:enumerable-sampling}. For a vector~$\hat y \in H^\star$,
  let~$C(\hat y)$ contain all index sets~$I \in {[n] \choose k}$ for which~$\hat y[I]$
  contains~$k$ distinct values. Let us extend this notation to sets~$H' \subseteq H^\star$
  via~$C(H') := \bigcup_{\hat{y} \in H'} C(\hat{y})$.
  Our goal is to find a set~$H' \subseteq H^\star$ of small size such that
  $C(H') = {[n] \choose k}$.

  \begin{claim}
    Let~$\mathcal F \subseteq {[n] \choose k}$ and~$n \geq 1.256\alpha k$.
    There exists a vector~$\hat y \in H^\star$ such that~$C(\hat y)$
    contains a fraction of at least $e^{-2\alpha k^2/n} {\alpha k \choose k} \frac{k!}{(\alpha k)^k}$ sets from~$\mathcal F$.
  \end{claim}
  \begin{proof}
    Fix any index set~$I \in \mathcal F$. Since~$|I| \leq k$ the random variables
    $\{\hat Y_i\}_{i \in I}$ are independent, the probability that a 
    vector~$\hat y \in H^\star$ chosen uniformly at random hits~$I$ is
    \[
      \mathbb{P}[I \in C(\hat y)] 
        \geq {\alpha k \choose |I|} k! \Big( \frac{1}{\alpha k} - \frac{1}{p} \Big)^{|I|}
        \geq (1-\alpha k/p)^{k} {\alpha k \choose k} \frac{k!}{(\alpha k)^k}.
    \]
    Accordingly, the expected number of sets in~$\mathcal F$ hit by at least one
    member of~$H^\star$ is at least
    $(1-\alpha k/p)^k {\alpha k \choose k} \frac{k!}{(\alpha k)^k} \cdot |\mathcal F|$
    and we conclude that at least one vector of~$H^\star$ must hit that
    many members of~$\mathcal F$. Since~$(1-c/x)^{x/c} \geq e^{-2}$ for~$x \geq 1.256c$
    we can bound the first term of this expression by
    \[
      \big(1 - \frac{\alpha k}{p})^k \geq e^{-2\alpha k^2 / p} \geq e^{-2\alpha k^2 / n}
    \]
    and arrive at the claimed bound.
  \end{proof}

  \noindent
  The above claim now gives us a method of constructing~$H'$ greedily:
  initialize~$H' = \emptyset$ and list the members~$\hat y_1, \hat y_2,
  \ldots$ of~$H^\star$. For each vector~$\hat y_i$, compute~$|C(\hat y_i)
  \setminus C(H')|$, that is, the number of index sets hit by~$\hat y_i$ that
  are not yet hit by any member of~$H'$. If 
  \[
    |C(\hat y_i) \setminus C(H')| 
      \geq e^{-2\alpha k^2 / n} {\alpha k \choose k} \frac{k!}{(\alpha k)^k}
            \Big({n \choose k} -|C(H')|\Big), 
  \]
  then add~$\hat y_i$ to~$H'$, otherwise drop~$\hat
  y_i$. In either case, proceed with~$\hat y_{i+1}$ until~$H^\star$ has been
  completely traversed.
  Note that in every step, we can compute the numbers~$|C(\hat y_i) \setminus
  C(H')|$ and $|C(H')|$ in time~$O\big({n \choose k} |H'| k \big)$ by simply
  enumerating all index sets and testing the vectors~$H' \cup \{\hat y_i\}$.

  Finally, let us bound the number of steps the above algorithm takes and
  therefore the size of~$H'$. Since every member added to~$H'$ covers at least
  a fraction of~$e^{-2\alpha k^2 / n} {\alpha k \choose k} k!/(\alpha k)^k$
  uncovered index sets, the number of steps~$t$ taken by the algorithm can be
  bounded by solving
  \[
    |\mathcal C| \Big(1 - e^{-2\alpha k^2 / n} {\alpha k \choose k} \frac{k!}{(\alpha k))^k} \Big)^t = 
    {n \choose k} \Big(1 - e^{-2\alpha k^2 / n} {\alpha k \choose k} \frac{k!}{(\alpha k))^k} \Big)^t < 1 
  \]
  for the variable~$t$. 
  \begin{claim}
    For~$t \geq e^{2 \alpha k^2/n} \frac{(\alpha k)^k}{{\alpha k \choose k} k!} k \ln 2n$ the
    above expression is smaller than one.
  \end{claim}
  \begin{proof}
    Let us substitute~$x = e^{2 \alpha k^2/n} \frac{(\alpha k)^k}{{\alpha k \choose k} k!}$.
    In particular, we assume~$t \geq x k \ln 2n$. Then the above expression becomes
    \begin{align*}
        {n \choose k} \Big(1 - \frac{1}{x} \Big)^t \leq 
        {n \choose k} \Big(1 - \frac{1}{x} \Big)^{x k \ln 2n}
        \leq {n \choose k} e^{-k \ln 2n} \leq \Big( \frac{e}{2k} \Big)^k
    \end{align*}
    where we used that~$(1-\frac{1}{x})^x \leq e^{-1}$. Clearly, the right hand side is smaller
    than one for~$k \geq 2 > e/2$ and the claim follows.
  \end{proof}

  \noindent
  A more manageable expression for the bound on~$t$ 
  (ignoring the small factor $e^{2 \alpha k^2/n} k \ln 2n$) can be derived using Stirling's
  approximation:
  \begin{align*}
      \phantom{{}\leq{}} \frac{(\alpha k)^k}{{\alpha k \choose k} k!} 
          &= (\alpha k)^k \frac{\big((\alpha -1)k \big)!}{(\alpha k)!} 
      \sim (\alpha k)^k \sqrt{1-1/\alpha} \Big( \frac{(\alpha-1)^{\alpha-1}}{\alpha^\alpha} \frac{e}{k} \Big)^k \\
      &= \sqrt{1-1/\alpha} \Big( (1-1/\alpha)^{\alpha-1} e \Big)^k
      = \Theta( \tradeoff(\alpha)^k ).
  \end{align*}
  We conclude that~$|H'| = O( \tradeoff(\alpha)^k k \log n )$.
  The bounds on time and space follow immediately from $\mathcal C(n,k) = H'$.
\end{proof}

\begin{definition}[Rounded splitter]
  Let~$\mathcal S$ be an indexed splitter. For rational numbers~$\tilde n,
  \tilde k,\tilde t$ we denote by~$[\mathcal S(\tilde n, \tilde k, \tilde t)]$
  the collection of indexed
  splitters~$\{ \mathcal S(n,k,t) \mid n \in \{\floor{\tilde n}, \ceil{\tilde n}\}, k \in \{\floor{\tilde k}, \ceil{\tilde k}\}, 
  t \in \{\floor{\tilde t}, \ceil{\tilde t}\},\}$.
\end{definition}

\noindent
We will treat the collection~$[\mathcal S(\tilde n, \tilde k, \tilde t)]$ like
an indexed splitter with the understanding that a query to~$[\mathcal S]$ carries,
besides the index~$i$, also the appropriate values for~$n$, $k$ and~$t$. Since
the additional overhead of constructing and querying~$\mathcal S$ differs only
by a constant factor, we will not further analyse this overhead in the following.

The following lemma closely follows the construction presented in Theorem~3 in~\cite{Splitters},
but adapted to construct a hash family with a flexible amount of colors.
\FR{We will in the following make use of the Iverson bracket notation~$\Iver{\phi}$ which evaluate
  to~$1$ if~$\phi$ is true and~$0$ otherwise. Also recall that we use the
  symbol $\any$ for variables whose value is arbitrary or unimportant in
  the current context. }

\begin{lemma}\label{lemma:wrapping}
  For~$\alpha \geq 1$, let~$\mathcal S(n,k,t)$ be an indexed $(n,k,t)$-splitter 
  with profile~$(f,t_\init, s_\init, t_\query, s_\query)$.
  Then for every~$\ell \leq k$ we can construct an indexed~$(n,k,t)$-splitter~$\hat{\mathcal S}(n,k,t)$
  with profile~$(\hat f,\hat t_\init, \hat s_\init, \hat t_\query, \hat s_\query)$ where
  \vspace*{-.5em}\begin{alignat*}{2}
    \hat f(n,k,t)        &=&~& k^{O(1)} {k^2 \choose \ell} f(k^2,\ceil{k/\ell},\ceil{t/\ell})^\ell \log n, \hspace*{3.35cm}  \\
    \hat t_\init(n,k,t)  &=&~& k^{O(1)} n \log n + O( t_\init(k^2, k/\ell, t/\ell) ),  \\
    \hat s_\init(n,k,t)  &=&~& k^{O(1)} n \log n + O( s_\init(k^2, k/\ell, t/\ell) ),  \\
    \hat t_\query(n,k,t) &=&~& O(\ell \, t_\query(k^2, \ceil{k/\ell}, \ceil{t/\ell}) + n),  \\
    \text{and}\quad\hat s_\query(n,k,t) &=&~& O(\ell \, s_\query(k^2, \ceil{k/\ell}, \ceil{t/\ell}) + \log n).
  \end{alignat*}
\end{lemma}
\begin{proof}
  In order to initialize~$\hat{\mathcal S}$, we
  construct the family~$\mathcal A := \mathcal A(n,k)$ from Proposition~\ref{prop:famA},
  the family~$\mathcal B := \mathcal B(k^2,k,\ell)$
  as well as the indexed splitter~$[\mathcal S] := [\mathcal S(k^2, k/\ell, t/\ell)]$
  in time
  \[
    \hat t_\init(n,k,t) = k^{O(1)} n \log n + O( t_\init(k^2, k/\ell, t/\ell) )
  \]
  using space
  \[
    \hat s_\init(n,k,t) = k^{O(1)} n \log n + O( s_\init(k^2, k/\ell, t/\ell) ),
  \]
  as claimed.

  In order to answer a query~$i$ for~$\hat S$, we decompose the query into
  $i = (i_a, i_b, i_1, \ldots, i_\ell)$ according to some indexing scheme (\eg
  choosing appropriate substrings of the bitwise representation of~$i$). 
  We then choose~$\bar a \in \mathcal A$ according to~$i_a$, a partition~$B^1 \cup \ldots \cup B^\ell$
  of~$[k^2]$ into~$\ell$ blocks according to~$i_b$ from~$\mathcal B(k^2, k, \ell)$
  and vectors~$\bar s_{i_1}, \ldots, \bar s_{i_\ell} \in [\mathcal S]$. We assume that
  the indexing scheme is such that~$\bar s_{i_1},\ldots,\bar s_{i_j}$ are taken
  from~$(k^2,\floor{k/\ell},\any)$-splitters and
  $\bar s_{i_{j+1}}, \ldots, s_{i_\ell}$ from from~$(k^2,\ceil{k/\ell},\any)$-splitters 
  where~$j$ is chosen such that
  \[
    \floor{k/\ell}j + \ceil{k/\ell}(\ell-j) = k,
  \]
  and further that~$\bar s_{i_1},\ldots,\bar s_{i_h}$ are taken from
  $(k^2,\any,\floor{t/\ell})$-splitters and~$\bar s_{i_{h+1}},\ldots,\bar s_{i_\ell}$
  from $(k^2,\any,\ceil{t/\ell})$-splitters where~$h$ is similarly chosen such that
  \[
    \floor{t/\ell}h + \ceil{t/\ell}(\ell-h) = t.
  \]
  The result~$\bar y$ of the query is now constructed as follows.
  For~$c \in [n]$, let~$B^p$ be the block of the chosen partition~$B_1 \cup \ldots \cup B^\ell$
  where~$\bar a[c] \in B^p$. Then~$\bar y[c] = \bar s_{i_p}[\bar a[c]] + \textsf{offset}(p)$,
  where
  \[
    \textsf{offset}(p) = \floor{t/\ell}p + \Iver{h > p}(h-p)
  \]
  shifts the colors taken from the vectors~$\bar s_{\any}$ in order to combine them
  without collisions. 

  Let us at this point convince ourselves that the vectors~$\bar y_i$ constructed this way indeed
  form a~$(n,k,t)$-splitter. Clearly, $\bar y_i \in [t]^n$ so it is left to consider
  the splitting property. To that end, fix an arbitrary index set~$I \in {[n] \choose k}$.
  First, there exists~$a \in \mathcal A$ such that~$\bar a[I]$ receives indeed~$k$ 
  distinct values~$C \in {[k^2] \choose k}$. For this subset~$C$, there exists a vector from
  $\mathcal B(k^2,k,\ell)$ that splits~$C$ evenly into parts of size~$\floor{k/\ell}$ and~$\ceil{k/\ell}$.
  For each such part~$C_i$, there now exists a vector~$\bar s_i \in [\mathcal S]$ such
  that~$\bar s_i[C_i]$ contains~$|C_i|$ distinct values. Since the values contained
  in the~$\bar s_{\any}$ are offset to avoid collisions, we conclude that the specific
  choices of~$\bar a$, $\mathcal B$, and~$\bar s_1 \ldots \bar s_\ell$ result in
  a vector~$\bar y$ such that~$\bar y[I]$ indeed contains~$k$ distinct values.

  In total, the time taken to answer a query is
  \[
    \hat t_\query(n,k,t) = O( \ell \, t_\query(k^2, \ceil{k/\ell}, \ceil{t/\ell})).
  \] 
  and it uses space
  \[
    \hat s_\query(n,k,t) = O(\ell \, s_\query(k^2, \ceil{k/\ell}, \ceil{t/\ell}) + \log n).
  \]
  The size of~$\hat S$ is computed easily by considering the size of~$\mathcal A$, $\mathcal B$ and the
  number of partitions of~$[k^2]$ into~$\ell$ parts:
  \[
    \hat f(n,k,t) = k^{O(1)} {k^2 \choose \ell} f(k^2,\ceil{k/\ell},\ceil{t/\ell})^\ell \log n. \qedhere
  \]
\end{proof}

\noindent
We finally arrive at the main theorem of this section.

\begin{theorem}\label{thm:polyspace-splitter}
  There exists a~$(n,k,\alpha k)$-splitter~$\doublehat{\mathcal S}$ with
  \begin{align*}
    \doublehat f(n,k,\alpha k) &= k^{O(1)}  \tradeoff(\alpha)^{k + o(k)} \log n, && \\
    \doublehat t_\init(n,k,\alpha k) &= k^{O(1)} \tradeoff(\alpha)^{o(k)} n \log n, 
    & \doublehat s_\init(n,k,\alpha k) &= k^{O(1)} n \log n, \\
    \doublehat t_\query(n,k,\alpha k) &= k^{O(1)} + O(n),
    & \doublehat s_\query(n,k,\alpha k) &= k^{O(1)} + O(\log n).
  \end{align*}
\end{theorem}
\begin{proof}
  We apply the construction of Lemma~\ref{lemma:wrapping} \emph{twice}, first 
  with~$\ell = k / \log^2 k$ and then with~$\ell = \log k$, to the splitter~$\mathcal S$
  from Lemma~\ref{lemma:new-splitter}. The resulting family has a size of
  \begingroup 
  \allowdisplaybreaks  
  \begin{align*}
      \doublehat f(n,k,\alpha k) 
    &= k^{O(1)} {k^2 \choose \log k} \hat f\big(k^2, k/\log k, \alpha k/\log k\big)^{\log k} \log n \\
    &= k^{O(1)} e^{o(k)} \log n \\
    &\quad~~ \cdot \Big[ 
            {(k/\log k)^2 \choose k/\log^2 k } f\big((k/\log k)^2, \log k, \alpha \log k\big)^{k/\log^2 k}
          2 \log k \Big]^{\log k}  \\
    &= k^{O(1)} e^{o(k)} f\big((k/\log k)^2, \log k, \alpha \log k\big)^{k/\log k} \log n \\
    &= k^{O(1)} e^{o(k)} \Big(  \tradeoff(\alpha)^{\log k} \, 2 \log \frac{k}{\log k} \Big)^{k/\log k} \log n \\
    &= k^{O(1)}  \tradeoff(\alpha)^{k + o(k)} \log n.
  \intertext{The remaining profile of~$\doublehat{\mathcal S}$ can be bounded as follows:}
    \doublehat t_\init(n,k,\alpha k)
    &= k^{O(1)} n \log n + O( \hat t_\init(k^2, k/\log k, \alpha k/\log k) ) \\
    &= k^{O(1)} n \log n + O( t_\init((k/\log k)^2, \log k, \alpha \log k) ) \\
    &= k^{O(1)} n \log n + 
        \tradeoff(\alpha)^{\log k} {(k/\log k)^2 \choose \log k} \big(2 (k\log k)^2\big)^{\alpha \log k} \log^{O(1)} \! k \\
    &= k^{O(1)} \tradeoff(\alpha)^{o(k)} n \log n. \\
    \doublehat s_\init(n,k,\alpha k)
    &= O( \hat s_\init(k^2, k/\log^2 k, \alpha k/\log^2 k) ) + k^{O(1)} n \log n  \\
    &= O( s_\init((k/\log k)^2, \log k, \alpha \log k) ) + k^{O(1)} n \log n  \\
    &= O( \tradeoff(\alpha)^{\log k} \log^2 (k/\log k) )  + k^{O(1)} n \log n = k^{O(1)} n \log n. \\
    \doublehat t_\query(n,k,\alpha k)
    &=  O(k \, \hat t_\query(k^2, k/\log k, \alpha /\log k) + n) \\
    &=  O\big(k^2 \, t_\query((k/\log k)^2, \log k, \alpha \log k) + n\big) \\
    &=  k^{O(1)} + O(n)  \\
    \doublehat s_\query(n,k,\alpha k)
    &= O(k \, \hat s_\query(k^2, k/\log k, \alpha k/\log k) + \log n) \\
    &= O\big(k^2 \, s_\query((k/\log k)^2, k/\log k, \alpha \log k) + k\log k + \log n\big) \\
    &= k^{O(1)} + O(\log n). \qedhere
  \end{align*}
  \endgroup
\end{proof}

\noindent 
Since the an indexed splitter with polynomial query time~$t_\query$
in particular means that we can enumerate its members with polynomial delay,
Theorem~\ref{thm:polyDelayHash} follows directly from Theorem~\ref{thm:polyspace-splitter}.

\section{Proof of Proposition \ref{prop:enrichCalc}}\label{sec:prop}

\noindent
As H\"{u}ffner \etal noted~\cite{ColorEngineering}, the running time of 
color coding algorithms can usually be improved by using more than~$k$ colors, balancing
the success probability of the coloring with the running time of the algorithm that uses
the coloring. In particular, they showed that a typical running time 
of~$\frac{t^{k}}{{t \choose k}\cdot k!} \cdot 2^t$ (were the first
term is the reciprocal of the success probability and the second term corresponds to an algorithm that needs
time~$2^t$ for an instance colored with $t$ colors) can be bounded by
\[
  \frac{(\alpha k)^{k}}{{\alpha k \choose k}\cdot k!} \cdot 2^{\alpha k} = \OO^*(4.314^k)
\]
for~$\alpha = 1.3$, a value that was derived numerically. Using the function~$\tradeoff(\alpha)$,
we can shed some light on how this value comes about. Consider a running time
\[ 
  \frac{(\alpha k)^{k}}{{\alpha k \choose k}\cdot k!} \cdot 2^{\alpha k} 
  \leq \tradeoff(\alpha)^k 2^{\alpha k} = (h(\alpha)\cdot e)^k, 
\]
 where $h(\alpha)=\big(1-1/{\alpha}\big)^{\alpha-1} 2^\alpha$. Let us minimize $h(\alpha)$ over $\alpha>1$. Since
\[
  \frac{{\rm d} h(\alpha)}{{\rm d} \alpha} 
  = (\alpha-1)^{\alpha-1} \Big( \frac{2}{\alpha} \Big)^{\alpha} \Big(\alpha \ln\Big( \frac{2(\alpha-1)}{\alpha}\Big) + 1\Big)
\]
and the first two terms of this expression have no roots for~$\alpha > 1$, we 
are left with solving
\begin{equation}\label{eq4}
  \alpha \ln\Big(\frac{2(\alpha-1)}{\alpha}\Big) = -1
  \iff 
  \Big(1 - 1/{\alpha} \Big)^\alpha 2^\alpha = \frac{1}{e} 
\end{equation}
for~$\alpha$. The only solution is $\alphaopt = 1.302017\dots$ and note that
 $\frac{{\rm d} h(\alpha)}{{\rm d} \alpha}$ is negative for $\alpha=1.2$ and
 positive for $\alpha=1.4$. Thus, $h(\alpha)\ (\alpha>1)$  is minimized for
 $\alpha=\alphaopt$. By (\ref{eq4}),
 $$h(\alphaopt)\Big(1-\frac{1}{\alphaopt}\Big)=\frac{1}{e}$$ and hence the
 optimal running time is \MZ{bounded from above} by
\[
  (h(\alphaopt)e)^k
  = \Big(\frac{\alphaopt}{\alphaopt-1} \Big)^k\le \Big(\frac{1.302018}{0.302017} \Big)^k = \OO^*(4.312^k).
\]
This proves Proposition \ref{prop:enrichCalc} for $p=0$ and thus for arbitrary $p$.

\section{Discussion}\label{sec:conclusion}

\noindent
In this paper, we presented \GG{an enhancement of color coding} to design polynomial-space
parameterized algorithms. We provided three applications centered around 
out-branchings, spanning trees, and perfect matchings. In particular, we obtained a
deterministic polynomial-space algorithm for {\sc $k$-IOB} that runs in time
\polyIOB. Along the way, we showed how to enumerate $(n,k,\alpha k)$-splitters
one by one with polynomial delay. In addition, we demonstrated that our
approach can be adapted to enable the development of faster parameterized
algorithms in case the use of exponential space is permitted.

Let us conclude our paper with an interesting open question. Let $\alpha_1(D)$
be the maximum size of a matching in a digraph $D$. By
Lemma~\ref{lemma:matching} (ii), $D$ has an $\alpha_1(D)$-internal out-branching. The following problem is natural: what is the parameterized
complexity of deciding whether $D$ has an  $(\alpha_1(D)+k)$-internal out-branching, where $k$ is the parameter?

\bigskip
{\noindent\bf Acknowledgment.} We thank Saket Saurabh for helpful information concerning splitters.


\end{document}